\RequirePackage[l2tabu,orthodox]{nag}
\documentclass[11pt,a4paper]{article}

\usepackage[T1]{fontenc}
\usepackage[utf8]{inputenc}
\usepackage{microtype}
\usepackage{lmodern}
\usepackage{enumerate}

\usepackage[margin=1in]{geometry}
\usepackage{amsmath,amsthm,amssymb}
\usepackage{filecontents}
\usepackage{color}
\usepackage{tikz}
\usetikzlibrary{calc}        
\usetikzlibrary{decorations.markings}

\usepackage{float}

\usepackage{hyperref}
\hypersetup{colorlinks=true,citecolor=blue, linkcolor=blue, urlcolor=blue}

\newcommand{\ceil}[1]{\lceil #1 \rceil}
\newcommand{\floor}[1]{\lfloor #1 \rfloor}

\newcommand{\bis}{\textnormal{\#BIS}}
\newcommand{\maxbis}{\textnormal{\#Max-BIS}}

\newcommand{\kbis}{\textnormal{\#}\textnormal{Size-BIS}}
\newcommand{\klbis}{\textnormal{\#}\textnormal{Size-Left-BIS}}
\newcommand{\klmaxbis}{\textnormal{\#}\textnormal{Size-Left-Max-BIS}} 
\newcommand{\apklmaxbis}[1]{\textnormal{Deg-{#1}-\#ApxSizeLeftMaxBIS}}
 
\newcommand{\IS}{\textnormal{IS}}
\newcommand{\MAXIS}{\textnormal{MIS}}

\newcommand{\LIS}[1]{\ensuremath{\textnormal{IS}_{\mbox{\tiny $#1$-left}}}}
\newcommand{\MAXLIS}[1]{\ensuremath{\textnormal{IS}_{\mbox{\tiny $#1$-left-max}}}}
\newcommand{\LRIS}[2]{\ensuremath{\mbox{IS}_{#1,\,#2}}}
\def\maxind{\mu}
\newcommand{\maxleft}[1]{\ensuremath{\mu_{\mbox{\tiny $#1$-left}}}}

\newcommand{\sub}{\textnormal{\#Induced-Coloured-Subgraph}$[\Delta]$}
\newcommand{\domset}{\textnormal{\#Size-Dominating-Set}}
\newcommand{\clique}{Size-Clique} 
\newcommand{\ckclique}{\textnormal{\#Size-Clique}}
\newcommand{\apklbis}[1]{\textnormal{Deg-{#1}-\#ApxSizeLeftBIS}} 
\newcommand{\apkbis}[1]{\textnormal{Deg-{#1}-\#ApxSizeBIS}}
\newcommand{\rainbow}{\textnormal{\#Size-Partitioned-Biclique}}

\newcommand{\NoClique}{\mbox{CL}}
\newcommand{\Wo}{\textnormal{W[1]}}
\newcommand{\Wt}{\textnormal{\#W[2]}}

\newcommand{\eps}{\varepsilon}

\newcommand{\R}{\mathbb{R}}
\newcommand{\N}{\mathbb{N}} 
\newcommand{\pr}{\mathbb{P}}
\newcommand{\E}{\mathbb{E}}

\newcommand{\cg}[1]{\mathcal{#1}} 
\newcommand{\subg}[2]{\mbox{Sub}_{#1}(#2)}

\def\RHPi{\ensuremath{\mathsf{\#RH}\Pi_1}}
\def\numP{\mbox{\#P}}

\newcommand{\Ind}[2]{\operatorname{\#Ind}(#1\rightarrow #2)}
\newcommand{\Hom}[2]{\operatorname{\#Hom}(#1\rightarrow #2)}
\newcommand{\SHom}[2]{\operatorname{\#Surj}(#1\rightarrow #2)}

\newtheorem{thm}{Theorem}
\newtheorem{lem}[thm]{Lemma}
\newtheorem{defn}[thm]{Definition}

\newtheorem*{prob*}{Problem}
\newtheorem*{remark*}{Remark}

\def\myprob#1#2#3{\noindent
  \textbf{Name:} #1.\\
  \textbf{Input:} #2.\\
  \textbf{Output:}  #3.\\}

\def\myfprob#1#2#3#4{\noindent
  \textbf{Name:} #1.\\
  \textbf{Input:} #2.\\
  \textbf{Output:}  #3.\\
  \textbf{Parameter: } #4.\\}

\begin{document}
\title{A Fixed-Parameter Perspective on \#BIS%
  \normalfont\footnote{Part of this work was done while the authors were
visiting the Simons Institute for the Theory of Computing. A preliminary version~\cite{CDFGL-old} of this paper appeared in IPEC 2017.}%
 }
\author{Radu Curticapean\thanks{%
Basic Algorithms Research Copenhagen and IT University of Copenhagen, Denmark. 
Supported by ERC grants PARAMTIGHT (No.~280152) and SYSTEMATICGRAPH (No.~725978) and VILLUM Foundation grant 16582 while working on this paper.},
Holger Dell\thanks{IT University of Copenhagen, Denmark.},
Fedor Fomin\thanks{University of Bergen, Norway.},\\
Leslie Ann Goldberg\thanks{University of Oxford, UK.
The research leading to these results has received funding from 
the European Research Council under the European Union's Seventh Framework Programme (FP7/2007-2013) ERC grant agreement no.\ 334828. The paper 
reflects only the authors' views and not the views of the ERC or the European Commission. The European Union is not liable for any use that may be made of the information contained therein.}, and
John Lapinskas\footnotemark[5].}
\date{}
\maketitle

\begin{abstract} 
The problem of (approximately) counting the independent sets of a bipartite graph
(\#BIS) is the canonical approximate counting problem 
that is complete in the intermediate complexity class
\RHPi. It is believed that \#BIS does not have an efficient approximation algorithm but
also that it is not NP-hard. We study the robustness of the intermediate
complexity of \#BIS\ by considering variants of the problem parameterised by the
size of the independent set. We map the complexity landscape
for three problems, with respect to exact computation and approximation and
with respect to conventional and parameterised complexity.
The three problems are counting independent sets of a given size, counting independent
sets with a given number of vertices in one vertex class and counting maximum
independent sets amongst those with a given number of vertices in one vertex class.
Among other things, we show that all of these problems are NP-hard to approximate within any polynomial ratio. 
(This is surprising because the corresponding problems without the size parameter
are complete in \RHPi, and hence are not believed to be NP-hard.)
We also show that the first problem is \#W[1]-hard to solve exactly but admits an FPTRAS, whereas the other two are W[1]-hard to approximate even within any polynomial ratio. Finally, we show that, when restricted to graphs of bounded degree, all three problems  have  efficient exact fixed-parameter algorithms.
\end{abstract}

\section{Introduction}

The problem of (approximately) counting the independent sets 
of a bipartite graph, called \#BIS, is one of the  most important problems in the
field of approximate counting.
This problem is known to be complete in the intermediate  complexity class \RHPi~\cite{DGGJ}.
Many approximate counting problems are equivalent in difficulty to \#BIS,
including those that arise in spin-system problems~\cite{BISPotts,pnaspaper} and
in other domains.
These problems
are not believed to have efficient approximation algorithms, but they are 
also not believed
to be NP-hard.

In this paper we study the robustness of
the intermediate complexity of \#BIS by considering 
relevant fixed parameters. 
It is already known that the complexity of \#BIS is unchanged
when the \emph{degree} of the input graph is restricted 
(even if it is restricted to be  at most~$6$) \cite{degreebis}
but there is an efficient approximation algorithm 
when 
a stronger degree restriction (degree at most~$5$) is 
applied, even to the vertices 
in just one of the parts of the vertex partition of  the bipartite graph \cite{pinyan}.

We consider variants of the problem parameterised by the \emph{size} of the independent set. 
We first show that all of the following problems are \#P-hard to solve exactly and NP-hard
to approximate \emph{within any polynomial factor}.
\begin{itemize}
  \item $\kbis$:
  Given a bipartite graph~$G$ and a non-negative integer $k$, count
  the 
  size-$k$ independent sets of~$G$.
    \item $\klbis$: Given a bipartite graph~$G$ with vertex partition $(U,V)$ and a non-negative integer $k$, count
     the independent sets of~$G$  
  that have $k$ vertices in~$U$, and   
  \item $\klmaxbis$: Given a bipartite graph $G$ with vertex partition $(U,V)$ and a non-negative integer~$k$,
  count the maximum independent sets amongst all independent sets
  of~$G$ with $k$ vertices in~$U$.   \end{itemize}
The NP-hardness of these approximate counting problems is surprising 
given that  
the corresponding problems without the parameter $k$
(that is, the problem \bis\ and also the problem \maxbis, 
which is the problem of  counting the \emph{maximum} independent sets of
a bipartite graph)  are
both complete in \RHPi, and hence are not believed to be NP-hard.
Therefore, it is the introduction of the parameter~$k$ that causes the hardness.

To gain  a more refined perspective on these problems, we also study them from the perspective of  parameterised complexity, taking 
the number of vertices, $n$, as the size of the input and~$k$ as the fixed parameter. 
Our results are summarised in Table~\ref{tab:results}, and stated in detail  later in the paper.
Rows~1 and~3 of the table correspond to the conventional (exact and approximate) 
setting that we have already discussed. Rows~2 and~4 correspond to the 
parameterised complexity setting, which we discuss next.
As is apparent from the table, we have mapped the complexity
landscape for the three problems in all four settings.

\begin{table}[h]
  \centering
  \renewcommand{\arraystretch}{1.8}
  \begin{tabular}{| p{1.2cm} | p{4.1cm} | p{4.1cm} | p{4.1cm} |}
    \hline 
    & \kbis & \klbis & \klmaxbis \\
    \hline 
    Exact poly & \#P-complete even in graphs of max-degree~$3$. (Thm~\ref{cor:kbis-exact})
               & \#P-complete even in graphs of max-degree~$3$. (Thm~\ref{cor:kbis-exact})
               & \#P-hard even in graphs of max-degree~$3$. (Thm~\ref{thmexklmaxbis})
    \\
    \hline
    Exact FPT & \#\Wo-complete. (Thm~\ref{thm:k-fpt-exact})
              & \Wt-hard. (Thm~\ref{thm:kl-fpt-exact})
              & \Wo-hard. (Thm~\ref{thm:kl-max-approx})
    \\ 
       & FPT for bounded-degree graphs. (Thm~\ref{thm:kbisbd}(i))
       & FPT for bounded-degree graphs. (Thm~\ref{thm:kbisbd}(ii))
       & FPT for bounded-degree graphs. (Thm~\ref{thm:kbisbd}(iii))
    \\
    \hline
    Approx poly & NP-hard to approximate within any polynomial factor. (Thm~\ref{thm:k-approx})
                & NP-hard to approximate within any polynomial factor. (Thm~\ref{thm:kl-approx})
                & NP-hard to approximate within any polynomial factor. (Thm~\ref{thm:kl-max-approx})
    \\
    \hline
    Approx FPT & Has FPTRAS. (Thm~\ref{thm:kbis-fptras})
               & \Wo-hard to approximate within any polynomial factor. (Thm~\ref{thm:kl-approx})
               & \Wo-hard to approximate within any polynomial factor. (Thm~\ref{thm:kl-max-approx})
    \\\hline
  \end{tabular}
  \caption{\label{tab:results}%
    A summary of our results.
    Each column corresponds to one of the three problems that we consider ($\kbis$, $\klbis$ and $\klmaxbis$), and each row corresponds to one of the four settings we consider (exact polynomial-time, exact FPT-time, approximate polynomial-time, and
    approximate FPT-time).
  }
\end{table}

In parameterised complexity,  the central goal is to determine whether  computational problems have fixed-parameter tractable (FPT) algorithms, that is, algorithms that run in time $f(k)\cdot n^{O(1)}$ for some computable function~$f$. 
Hardness results are presented using the 
$W$-hierarchy~\cite{FG}, and in particular using the 
complexity classes
W[1] and W[2], which constitute the first two levels of the hierarchy.
It is known  (see \cite{FG}) that
  $\mbox{FPT} \subseteq \mbox{W}[1] \subseteq \mbox{W}[2]$
  and
these classes are widely believed to be distinct from each other.
It is also   known~\cite[Chapter~14]{FK} that the Exponential Time Hypothesis (see~\cite{IP}) implies 
$\mbox{FPT} \ne \mbox{W[1]}$.
Analogous classes $\mbox{\#W[1]}$ and $\mbox{\#W[2]}$ exist for counting problems~\cite{FG04}.

As can be seen from the table, we prove that all of our problems are at least \Wo-hard to solve exactly, which indicates 
(subject to the complexity assumptions in the previous paragraph)
that they cannot be solved by FPT algorithms. Moreover, \klbis\ and \klmaxbis\ are \Wo-hard to solve even approximately. 
It is known \cite{Mueller}
that each parameterised counting problem in the class 
\#W[i] has a randomised FPT approximation algorithm using a W[i] oracle,
so  W[i]-hardness is the appropriate hardness  notion for parameterised approximate counting problems.
By contrast, we show that
  \kbis\ \emph{can} be solved approximately in FPT time. In fact, it has
  an FPT randomized
approximation scheme~(FPTRAS).

Motivated by the fact that \bis\ is known to be \#P-complete to solve exactly even on graphs of degree 3~\cite{XIA2007111}, we also consider the case where the input graph has bounded degree. 
While the conventional problems remain intractable in this setting (row one of the table),
we prove that all three of our problems admit linear-time fixed-parameter algorithms for bounded-degree instances
(row two). Note that Theorem~\ref{thm:kbisbd}(i) is also implicit in independent work by Patel and Regts~\cite{DBLP:journals/corr/PatelR16}.
 
\section{Preliminaries}

For a positive integer~$n$, we let $[n]$ denote the set $\{1,\dots,n\}$.
We consider graphs~$G$ to be undirected.
For a vertex set $X\subseteq V(G)$, denote by $G[X]$ the subgraph induced by~$X$.
For a vertex $v\in V(G)$, we write $\Gamma(v)$ for its open neighbourhood (that is, excluding $v$ itself).

Given a graph~$G$, we denote the size of a maximum independent set of~$G$ by
$\maxind(G)$.
We denote the number of all independent sets of~$G$ by $\IS(G)$, the number of size-$k$
independent sets of~$G$ by $\IS_k(G)$, and the number of size-$\maxind(G)$ independent
sets of~$G$ by $\MAXIS(G)$.
A bipartite graph~$G$ is presented as a triple $(U,V,E)$  in which  $(U,V)$ is a partition
of the vertices of $G$ into \emph{vertex classes}, and  
$E$ is a subset of $U\times V$.
 If $G=(U,V,E)$ is a bipartite graph then 
an independent set~$S$ of $G$ is said to be an ``$\ell$-left independent set of~$G$''
if $|S \cap U| = \ell$.
The size of a \emph{maximum} $\ell$-left independent set of~$G$ is denoted by~$\maxleft{\ell}(G)$.
An $\ell$-left independent set of 
$G$ is said to be ``$\ell$-left-maximum'' if and only if it has size $\maxleft{\ell}(G)$.
Finally,
$\LIS{\ell}(G)$ denotes
the number of $\ell$-left independent sets  of $G$  
and $\MAXLIS{\ell}(G)$ denotes the number of $\ell$-left-maximum independent sets of $G$.
Using these definitions, we now give formal definitions of \bis\ and of the three problems that we study.

 \bigskip
 
  \myprob{\bis}{A bipartite graph $G$}{$\IS(G)$}

\myfprob{\kbis}{A bipartite graph $G$ and a non-negative integer $k$} {$\IS_k(G)$}{$k$}

\myfprob{\klbis}{A bipartite graph $G$ and a non-negative integer $\ell$}{$\LIS{\ell}(G)$}{$\ell$}

\myfprob{\klmaxbis}{A bipartite graph $G $ and a non-negative integer $\ell$}{$\MAXLIS{\ell}(G)$}{$\ell$}

For each of our computational problems, we add ``[$\Delta$]'' to the end of the name of the problem
to indicate that the input graph~$G$
 has degree at most~$\Delta$.
For example, $\bis[\Delta]$
is the problem defined as follows.\medskip

    \myprob
      {$\bis[\Delta]$}
      {A bipartite graph~$G$ with degree at most~$\Delta$}
      {$\IS(G)$}

When stating quantitative bounds on running times of algorithms, we assume the standard word-RAM machine model with logarithmic-sized words.

\section{Exact computation: Hardness results}
\label{sec:exacthard}

In this section, we prove the hardness results presented in the first two rows
of Table~\ref{tab:results}.

\subsection{Polynomial-time computation}

We prove that all three problems are $\numP$-hard,
even when the input graphs are restricted to have degree at most~$3$.

\begin{thm}\label{cor:kbis-exact}
$\kbis[3]$ and $\klbis[3]$ are  both \textnormal{\#P}-complete.
\end{thm}
\begin{proof}
The problems $\kbis[3]$ and $\klbis[3]$ are in \#P, which can be deduced
from their definitions.
We show that the problems are \#P-hard.
Xia, Zhang and Zhao~\cite[Theorem 9]{XIA2007111} show that $\bis[3]$ is
\#P-hard, even 
under the additional restriction that the input graph is planar and $3$-regular.

There is a straightforward reduction from $\bis[3]$ to $\kbis[3]$.
Suppose that $G$ is an $n$-vertex input to $\bis[3]$. 
Then $\IS(G) = \sum_{k=0}^n \IS_k(G)$. 
Using an oracle for  $\kbis[3]$ 
(with the graph $G$ and each $k\in \{0,\ldots,n\}$) one can therefore
  compute $\IS(G)$, as desired.

Similarly, there is a straightforward reduction from $\bis[3]$ to $\klbis[3]$, using the
fact that 
$\IS(G) = \sum_{\ell=0}^n \LIS{\ell}(G)$. 
Thus,  the problems $\kbis[3]$ and $\klbis[3]$ are both  \#P-hard.
\end{proof}

\begin{thm}\label{thmexklmaxbis}
$\klmaxbis[3]$ is \textnormal{\#P}-hard.
\end{thm}
\begin{proof}
Vadhan has shown \cite[Corollary 4.2(1)]{Vadhan} that
$\maxbis[3]$ is \#P-complete. We now 
reduce $\maxbis[3]$ to 
$\klmaxbis[3]$.   Let $G=(U,V,E)$ be an instance of \maxbis[3] and let $s=|U|$.  For 
each $j\in \{0,\ldots,s\}$,
 let $x_j$ be the number of 
 size-$\maxind(G)$
 $(s-j)$-left independent sets of~$G$.
 We wish to compute $\MAXIS(G) = \sum_{j=0}^{s} x_j$, so it suffices to show how to compute 
 the vector $(x_0,\ldots,x_s)$   
 in polynomial time
 using an oracle for $\klmaxbis[3]$  --- this is what we do in the remainder of the proof.

For every  non-negative integer~$i$,  let $G_i = (U_i, V_i, E_i)$ 
be the graph formed from~$G$ by adding a disjoint matching of size $s+i$. 
Note that $\maxind(G_i) = \maxind(G)+s+i$.
Also, $G_i$ has an $s$-left independent set
of size $\maxind(G_i)$ (to see this, consider any size-$\mu(G)$ independent set of~$G$, say one that is
$a$-left for some $a \in \{0,\ldots,s\}$, and augment this with $s-a$ matching vertices from~$U_i$ and
$ i+a$ matching vertices from~$V_i$).
Let
$w_i$ be the number of 
size-$(\maxind(G_i))$ 
$s$-left independent sets of~$G_{i}$
and note that  $\MAXLIS{s}(G_i) = w_i$.  
Since 
  $G_i$ has degree at most~$3$,
$w_{i}$ can be computed using an oracle for $\klmaxbis[3]$
(using the input graph $G_i$ and setting the input~$\ell$ equal to~$s$).

From the definitions of~$x_j$ and $w_i$, we have
\begin{equation}\label{eq:simple}
w_i = 
\sum_{j=0}^{s} x_{j} \binom{s+i}{j}.
\end{equation}

Now let $M$ be the matrix whose rows and columns are indexed by
$\{0,\ldots,s\}$ and 
whose entry $M_{i,j}$ is $\binom{s+i}{j}$.
Let $\boldsymbol{w}$ be the transpose of the row vector $(w_0,\ldots,w_s)$
and $\boldsymbol{x}$ be the transpose of the row vector $(x_0,\ldots,x_s)$.  Then Equation~\eqref{eq:simple} can be re-written
as $\boldsymbol{w} = M \boldsymbol{x}$. 

Now \cite[Corollary 2]{Gessel}
shows that $M$ is invertible 
(taking $k=s+1$, 
$a_i = s+i-1$ and
$b_{j}=j-1$ for $1\leq i,j \leq k$, in the language of the corollary), so 
the vector $\boldsymbol{x}$
can be computed as $ \boldsymbol{x} = M^{-1} \boldsymbol{w} $. 
Since it suffices to compute~$\boldsymbol{x}$, 
and the vector $\boldsymbol{w}$ can be computed using the oracle,
this completes the reduction. 
\end{proof}

\subsection{Fixed-parameter intractability}\label{sec:fpt-exact-hard}

We first define the parameterised complexity classes relevant in this paper, 
namely, the class \Wo\ of decision problems, and the counting classes \#\Wo\ and \Wt. 
For simplicity, we do so in terms of complete problems and reductions. 
The following definitions are taken from Flum and Grohe~\cite{FG}.

\begin{defn}
Let $F$ and $G$ be parameterised problems. For any instance $x$ of $F$, write $k(x)$ for the parameter of $F$ and $|x|$ for the size of $x$. For any instance $y$ of $G$, write $\ell(y)$ for the parameter of $y$. An \emph{FPT Turing reduction} from $F$ to $G$ is an algorithm with an oracle for $G$ that, for some computable functions $f,g:\N\rightarrow\N$ and for some constant $c \in \N$, solves any instance $x$ of $F$ in time at most $f(k(x))\cdot |x|^c$ in such a way that for all oracle queries the instances $y$ of $G$ satisfy $\ell(y) \le g(k(x))$. 

Now, write $\mathcal{F}$ for the set of all instances of $F$, and for all $x \in \mathcal{F}$ write $F(x)$ for the desired output given input $x$. Likewise, write $\mathcal{G}$ for the set of all instances of $G$, and for all $y \in \mathcal{G}$ write $G(y)$ for the desired output given input $y$. Suppose $R:\mathcal{F}\rightarrow\mathcal{G}$ is a function satisfying the following properties: for all $x \in \mathcal{F}$, $F(x) = G(R(x))$; there exists a computable function $f:\N\rightarrow\N$ and a constant $c \in \N$ such that for all $x \in \mathcal{F}$, $R(x)$ is computable in time at most $f(k(x))\cdot |x|^c$; there exists a computable function $g:\N\rightarrow\N$ such that for all $x \in \mathcal{F}$, $\ell(R(x)) \le g(k(x))$. 
If $F$ and $G$ are decision problems, we call $R$ an \emph{FPT many-one reduction} from $F$ to $G$; if $F$ and $G$ are counting problems, we call $R$ an \emph{FPT parsimonious reduction} from $F$ to $G$.
\end{defn}

We define \Wo\ in terms of the following problem.\medskip

\myfprob{\clique}{A graph $G$ and a positive integer $k$}{True if $G$ contains a $k$-clique, false otherwise}{$k$}

\noindent\Wo\ is the set of all parameterised decision problems with an FPT many-one reduction to \clique. We say a problem $F$ is \Wo-hard if there is an FPT Turing reduction from \clique\ to $F$. For a proof that this is equivalent to the standard definition of \Wo, see e.g.,\ Downey and Fellows~\cite[Theorem 21.3.4]{Downey}. 

We define \#\Wo\ in terms of the following problem.\medskip

\myfprob{\ckclique}{A graph $G$ and a positive integer $k$}{The number of $k$-cliques in $G$}{$k$}

\noindent\#\Wo\ is the set of all parameterised counting problems with an FPT parsimonious reduction to \ckclique. We say a problem $F$ is \#\Wo\emph{-hard} if there is an FPT Turing reduction from \ckclique\ to $F$. For a proof that this is equivalent to the standard definition of \#\Wo, see e.g.,\ \cite[Theorem 14.18]{FG}. 

Recall that a set $D\subseteq V(G)$ is called a \emph{dominating set} of a graph~$G$ if every vertex~$v\in V(G)$ is either contained in~$D$ or adjacent to a vertex of~$D$. We define \Wt\ in terms of the following problem.\medskip

\myfprob{\domset}{A graph $G = (U,E)$ and a positive integer $k$}{The number of dominating sets of $G$ of size $k$}{$k$}

\noindent\Wt\ is the set of all parameterised counting problems with an FPT parsimonious reduction to \domset. We say a problem $F$ is \Wt-hard if there is an FPT Turing reduction from \domset\ to $F$. For a proof that this is equivalent to the standard definition of \Wt, see~\cite[Theorem 19]{FG04}).

In order to prove our exact fixed-parameter hardness results, we 
consider the following problem.

\medskip

\myfprob{\rainbow}
{An integer $t$, a $2t$-coloured graph $\cg{G}$, and a $2t$-coloured balanced biclique $\cg{H}$ on $2t$ vertices (i.e.\ a $2t$-coloured copy of $K_{t,t}$) in which every colour appears exactly once}
{The number $\subg{\cg{G}}{\cg{H}}$ of subgraphs $\cg{K} \subseteq \cg{G}$ with $\cg{K} \simeq \cg{H}$}
{$t$}

\begin{thm}\label{thm:k-fpt-exact}
	\kbis\ is \textnormal{\#W[1]}-complete.
\end{thm}
\begin{proof}
	We will prove first easiness, then hardness.
	
	\medskip\noindent\textbf{\#Size-BIS is in \#W[1]:} We give an FPT parsimonious reduction to \ckclique. Indeed, given an instance $(G,k)$ of \kbis\ with $G = (U,V,E)$, let $V' = U \cup V$, $E' = \{\{u,v\}\mid u,v \in V', (u,v) \notin E\}$, and $G' = (V',E')$. Then the size-$k$ independent sets of $G$ correspond exactly to the size-$k$ cliques of $G'$, as required.
	
	\medskip\noindent\textbf{\#Size-BIS is \#W[1]-hard:} We give an FPT Turing reduction from \rainbow. Note that the class $\{K_{t,t} \colon t \ge 1\}$ of all balanced bicliques is recursively enumerable and contains graphs of arbitrarily high treewidths, so \rainbow\ is \#W[1]-hard by a result of Curticapean and Marx~\cite[Theorem II.8]{CM}.

	Let $(t,\cg{G},\cg{H})$ be an instance of \rainbow. Write $\cg{G} = ((V,E),c)$. Without loss of generality, suppose the colours $\{1, \dots, t\}$ appear in one vertex class of $\cg{H}$ and the colours $\{t+1, \dots, 2t\}$ appear in the other. Let 
	\[E' = \{\{u,v\} \mid u,v \in V,\, c(u) \in [t],\, c(v) \in [2t] \setminus [t],\, \{u,v\} \notin E\}.\]
	Define a coloured graph $\cg{G}' = ((V, E'),c)$. Then each copy of $\cg{H}$ in $\cg{G}$ spans an independent set in $\cg{G}'$ in which every colour appears exactly once and vice versa, so $\subg{\cg{G}}{\cg{H}}$ is precisely the number of such independent sets in $\cg{G}'$.
	
	For any set $S \subseteq [2t]$, let $\mathcal{I}_S$ be the set of size-$2t$ independent sets in $\cg{G}'$ which contain no vertices with colours in $S$. By the inclusion-exclusion principle,
	\begin{equation}\label{eqn:rainbow}
		\subg{\cg{G}}{\cg{H}} = \left|\mathcal{I}_{\emptyset} \setminus \bigcup_{i=1}^{2t} \mathcal{I}_{\{i\}}\right| = |\mathcal{I}_{\emptyset}| - \sum_{\emptyset \ne S \subseteq [2t]} (-1)^{|S|-1}|\mathcal{I}_S|.
	\end{equation}
	Moreover, for any set $S \subseteq [2t]$, let $G_S$ be the bipartite graph $(U_S,V_S,E_S)$ defined by
	\begin{align*}
		U_S &= \{v\in V \mid c(v) \in [t] \setminus S\},\\
		V_S &= \{v\in V \mid c(v) \in [2t] \setminus ([t] \cup S)\},\\
		E_S &= \{(u,v) \mid u \in U_S, v \in V_S, \{u,v\} \in E'\}.
	\end{align*}
	Then $\mathcal{I}_S$ is precisely the set of size-$2t$ independent sets in $G_S$. Our algorithm therefore determines each $|\mathcal{I}_S|$ by calling a \kbis\ oracle with input $(G_S,2t)$, then uses \eqref{eqn:rainbow} to compute $\subg{\cg{G}}{\cg{H}}$.
\end{proof}

Next, we turn to the exact parameterised complexity of $\klbis$.
The hardness result we obtain for this problem is a bit stronger than for
$\kbis$: we prove that it is $\Wt$-hard.

\begin{thm}\label{thm:kl-fpt-exact}
\klbis\ is $\Wt$-hard.
\end{thm}
\begin{proof}
  We reduce from the dominating set problem.
  Let $G = (U,E)$ and $k$ be given as input for \domset\ where $U = \{u_1,\ldots,u_n\}$.
  The reduction computes the bipartite split graph of~$G$; formally, let $V =
  \{v_1,\ldots,v_n\}$, let $E' = \{(u_a,v_b) \mid \mbox{$a=b$ or $\{u_a,u_b\}\in
E$}\}$, and let $G' = (U,V,E')$.

For non-negative integers $\ell$ and $r$, we define an \emph{$(\ell,r)$-set of $G'$} to be  
a size-$\ell$ subset~$X$ 
of $U$ 
that has exactly $r$ neighbours in~$V$.
Let $Z_{\ell,r}$ be the number of $(\ell,r)$-sets of $G'$. Note that 
a size-$k$ subset $X$ of~$U$
is a dominating set of $G$ if and only if 
it is a $(k,n)$-set of $G'$, so there are precisely $Z_{k,n}$ size-$k$ dominating sets of~$G$.
  
The algorithm applies polynomial interpolation to determine $Z_{k,r}$ for all $r \in \{0,\dots,n\}$. 
We use a special case of the cloning construction from the proof of
Theorem~\ref{thm:k-fpt-exact}.
For every positive integer~$i$,
let $V_i = V \times [i]$, let 
$E_i' = \{(u,(v,b)) \in U \times V_i \mid (u,v) \in E'\}$, 
and let $G_i' = (U,V_i,E_i')$.  
For each $(k,r)$-set $X$ of $G'$, there 
are exactly $2^{i(n-r)}$ $k$-left independent sets $S$ of~$G'_i$
with $S\cap U = X$.  
Thus for all $i \in [n+1]$,
\begin{equation}\label{eqn:kl-fpt-exact}
\LIS{k}(G'_i) = \sum_{r=0}^n2^{i(n-r)}Z_{k,r}.
\end{equation}
  
Let $M$ be the $(n+1)\times (n+1)$ matrix whose rows are indexed by $[n+1]$ and
columns are indexed by $\{0, \dots, n\}$ such that $M_{i,r} = 2^{i(n-r)}$ holds.
Then \eqref{eqn:kl-fpt-exact} can be viewed as a linear
equation system $\boldsymbol{w} = M\boldsymbol{z}$, where
$\boldsymbol{w} = (\LIS{k}(G'_1), \dots, \LIS{k}(G'_{n+1}))^T$ and
$\boldsymbol{z} = (Z_{k,0}, \dots, Z_{k,n})^T$.
The oracle for \klbis\ can be used to compute $\boldsymbol{w}$, and 
$M$ is invertible since it is a (transposed) Vandermonde matrix.
Thus the reduction can compute $\boldsymbol{z}$, and in particular $Z_{k,n}$, as
required.
\end{proof}

We defer the proof of the \Wo-hardness of \klmaxbis\ to the next section, as it
is implied by the corresponding approximation hardness result.

\section{Approximate computation: Hardness results}

In this section, we prove the hardness results in rows~3 and~4 of
Table~\ref{tab:results}.
Note that the reductions
from Section~\ref{sec:exacthard} cannot be used  here, since \bis\ is not known to be NP-hard to approximate.
In order to state our hardness results formally, we introduce approximation
versions of the problems that we consider.\medskip

\myfprob{\apklmaxbis{$c$}}{A bipartite graph $G$ on $n$ vertices and a non-negative integer $\ell$}{A number $z$ such that $n^{-c}\cdot\MAXLIS{\ell}(G) \le z \le n^c\cdot\MAXLIS{\ell}(G)$}{$\ell$}

\myfprob{\apklbis{$c$}}{A bipartite graph $G$ on $n$ vertices and a non-negative integer $\ell$}{A number $z$ such that $n^{-c}\cdot\LIS{\ell}(G) \le z \le n^c\cdot\LIS{\ell}(G)$}{$\ell$}

\myfprob{\apkbis{$c$}}{A bipartite graph $G$ on $n$ vertices and a non-negative integer $k$}{A number $z$ such that $n^{-c}\cdot\IS_k(G) \le z \le n^{c}\cdot\IS_k(G)$}{$k$}

We first prove the results in the last column of Table~\ref{tab:results} and
establish the others by reduction.
\begin{thm}\label{thm:kl-max-approx}
  For all $c \ge 0$, \apklmaxbis{$c$} is both \textnormal{NP}-hard and \textnormal{\Wo}-hard.
\end{thm}
\begin{proof}
  Let $c$ be any non-negative integer. We will give a reduction from \clique\ to
  \apklmaxbis{$c$} which is both an FPT Turing reduction and a polynomial-time
  Turing reduction.
  The theorem then follows from the fact that \clique\ is both NP-hard
  \cite[Theorem~7.32]{Sipser}) and \Wo-hard \cite[Theorem~21.3.4]{Downey}.

  Let $(G,k)$ be an instance of \clique{} with $G = (V,E)$ and $n = |V|$.
  We use a standard powering construction to produce an intermediate instance
  $(G',k)$ of \clique{} with $G'=(V',E')$.
  More precisely, let $t = n^{2c}$, let $C$ be a set of $k$ new vertices, and
  let $V' = (V \times [t]) \cup C$.
  We define $E'$ such that
  \[
  E' = \big\{\{(u,i),(v,j)\}\mid \{u,v\} \in E, i,j \in [t]\big\} \cup
  \big\{\{u,v\} \mid u,v \in C,u\ne v\big\}\,.
  \]

  From $(G',k)$, we construct an instance $(G'',\ell)$ of \apklmaxbis{$c$} with
  $G''=(U,V',E'')$ and $\ell=\binom{k}{2}$.
  For this, let $U = \{u_e \mid e \in E'\}$ be a set of vertices and
  let $E'' = \{(u_e,v) \mid e\in E',\,v \in e\}$.
  The reduction queries the oracle for $(G'',\ell)$, which yields an approximate
  value~$z$ for the number $\MAXLIS{\ell}(G'')$.
  If $z\le n^c$, the reduction returns `no', there is no $k$-clique in~$G$, and
  otherwise it returns `yes'.
  It is obvious that the reduction runs in polynomial time.

  It remains to prove the correctness of the reduction.
  Let $\NoClique_k(G)$ be the number of $k$-cliques in $G$.
  The $\ell$-left-maximum independent sets $X$ of $G''$ correspond bijectively
  to the size-$\ell$ edge sets $\{e \mid u_e \in X \cap U\}$ of $G'$ which span
  a minimum number of vertices. Note that any set of $\ell = \binom{k}{2}$ edges
  must span at least $k$ vertices, with equality only in the case of a
  $k$-clique. Since~$G'$ contains at least one $k$-clique (induced by $C$), we
  have $\MAXLIS{\ell}(G'') = \NoClique_k(G')$. Moreover, each $k$-clique $X$ in
  $G$ corresponds to a 
  size-$t^k$ family of  
  $k$-cliques in $G'$.
  Each $k$-clique in the family consists of exactly one vertex from each set
  $\{x\} \times [t]$ such that $x \in V(X)$. This accounts for all $k$-cliques
  in $G'$ except
  $G'[C]$. Thus
  \begin{equation}\label{eqn:kl-max-approx-1}
  \MAXLIS{\ell}(G'') = \NoClique_k(G') = t^k\NoClique_k(G)+1.
  \end{equation}
  
  Let $z$ be the result of applying our oracle to $(G'',\ell)$.
  If $G$ contains no $k$-cliques, then by~\eqref{eqn:kl-max-approx-1} we have $z
  \le n^c\cdot\MAXLIS{\ell}(G'') = n^c$ and the reduction returns `no'.
  Otherwise, we have $z \ge n^{-c}\cdot\MAXLIS{\ell}(G'') \ge 
  n^{-c}(t^k+1) > n^c$
  and the reduction returns `yes'.
  Thus the reduction is correct and the theorem follows.
\end{proof}
 
\begin{thm}\label{thm:kl-approx}
  For all $c \ge 0$, \apklbis{$c$} is both \textnormal{NP}-hard and \textnormal{\Wo}-hard.
\end{thm}
\begin{proof}
  Let $c$ be any non-negative integer. We will give a reduction from the problem \apklmaxbis{$(c+1)$} to the problem \apklbis{$c$} which is both an FPT Turing reduction and a polynomial-time Turing reduction.
  The result then follows by Theorem~\ref{thm:kl-max-approx}.
  
  Let $(G,\ell)$ be an instance of \apklmaxbis{$c$}. Write $G = (U,V,E)$, let $n=|V(G)|$, and let $t = 6n$. Without loss of generality, suppose $n \ge 5$ and that $n$ is sufficiently large that $n^c2^{-n} \le 1$. Let $V' = V \times [t]$, let $E' = \{(u,(v,i)) \mid (u,v) \in E,\, i \in [t]\}$, and let $G' = (U,V',E')$. Let $\mu = \maxleft{\ell}(G)$, and let $z$ be the result of applying our oracle to $(G',\ell)$.
   
  For any non-negative integers $i$ and $j$, we define $\LRIS{i}{j}(G)$ to be the number of independent sets $X \subseteq V(G)$ with $|X \cap U| = i$ and $|X \cap V| = j$. Each $\ell$-left independent set $X$ of $G$ corresponds to the family of $\ell$-left independent sets of $G'$ consisting of $X \cap U$ together with at least one vertex from each set $\{x\} \times [t]$ such that $x \in X \cap V$. Thus by the definition of $\mu$,
  \begin{equation}\label{eqn:kl-approx}
    \LIS{\ell}(G') = \sum_{r=0}^{\mu-\ell}\LRIS{\ell}{r}(G)(2^t-1)^{r}.
  \end{equation}
  Since $G$ contains at most $2^n$ independent sets and $\LRIS{\ell}{\mu-\ell}(G) \ge 1$, we have
  \begin{equation*}
    (2^t-1)^{\mu-\ell} \le \LIS{\ell}(G') \le 2^n(2^t-1)^{\mu-\ell}.
  \end{equation*}
  Since $n^c \le 2^n \le (2^t-1)^{1/5}$, it follows that $(2^t-1)^{\mu-\ell-1/5} \le z \le (2^t-1)^{\mu-\ell+2/5}$, and hence the algorithm can obtain $\mu$ by rounding $\ell+\lg(z)/\lg(2^t-1)$ to the nearest integer. Moreover, by \eqref{eqn:kl-approx} we have
  \begin{align*}
  \LIS{\ell}(G') 
  &\le \LRIS{\ell}{\mu-\ell}(G)(2^t-1)^{\mu-\ell} + 2^{n}(2^t-1)^{\mu-\ell-1} \le 2\LRIS{\ell}{\mu-\ell}(G)(2^t-1)^{\mu-\ell}.
  \end{align*}
  It follows that $\LRIS{\ell}{\mu-\ell}(G)\le \LIS{\ell}(G')/(2^t-1)^{\mu-\ell} \le 2\LRIS{\ell}{\mu-\ell}(G)$, and hence that 
  \[
  n^{-c-1}\LRIS{\ell}{\mu-\ell}(G)\le z/(2^t-1)^{\mu-\ell} \le n^{c+1}\LRIS{\ell}{\mu-\ell}(G).
  \]
  The algorithm therefore outputs $z/(2^t-1)^{\mu-\ell}$.
\end{proof}

The following well-known Chernoff bound appears in e.g.,\ Janson, {\L}uczak and Rucinski~\cite[Corollary 2.3]{JLR2000}.
\begin{lem}\label{lem:chernoff}
  If $X \sim \textnormal{Bin}(n,p)$ is a binomial variable and $0 < \eps \le 3/2$, then
  \[\pr(|X - \E(X)| \ge \eps\E(X)) \le 2e^{-\eps^2\E(X)/3}.\]
  \qed
\end{lem}

  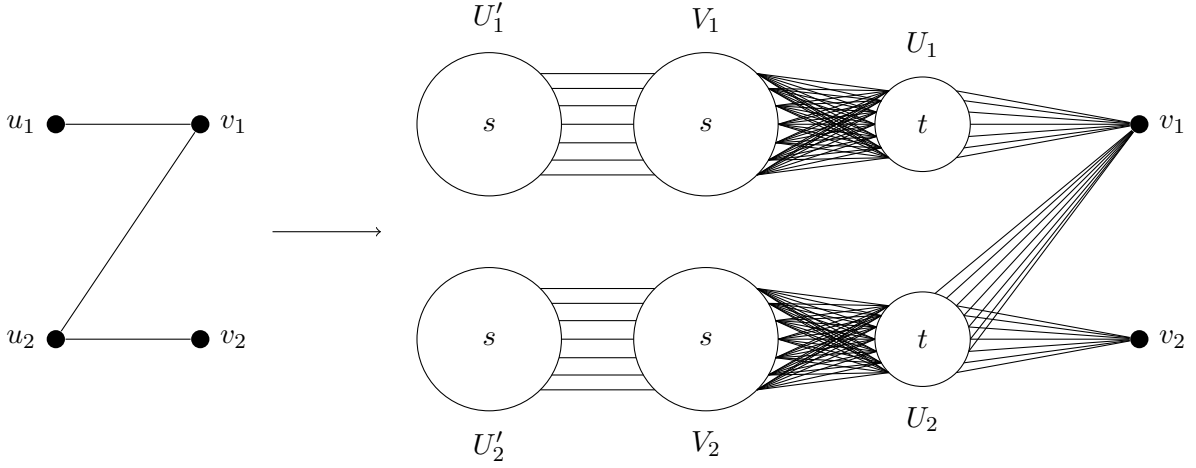
\begin{figure}
    \centering
      \begin{tikzpicture}[scale=0.95]
        \tikzstyle{vertex} = [circle, inner sep = 0.85mm, fill=black]
        \node (u1)  at (0,3) [vertex, label = left: $u_1$] {};
        \node (u2)  at (0,0) [vertex, label = left: $u_2$] {};
        \node (v1)  at (2,3) [vertex, label = right:$v_1$] {};
        \node (v2)  at (2,0) [vertex, label = right:$v_2$] {};
        \node (U1') at (6,3) {$s$};
        \node (U2') at (6,0) {$s$};
        \node (V1)  at (9,3) {$s$};
        \node (V2)  at (9,0) {$s$};
        \node (U1)  at (12,3) {$t$};
        \node (U2)  at (12,0) {$t$};
        \node (v1') at (15,3) [vertex, label=right:$v_1$] {};
        \node (v2') at (15,0) [vertex, label=right:$v_2$] {};
        
        \def \bigr{1}
        \def \smlr{0.66}
        
        \draw (U1') circle [radius=\bigr];
        \draw (U2') circle [radius=\bigr];
        \draw (V1)  circle [radius=\bigr];
        \draw (V2)  circle [radius=\bigr];
        \draw (U1)  circle [radius=\smlr];
        \draw (U2)  circle [radius=\smlr];
        
        \draw (u1) -- (v1) -- (u2) -- (v2);

        \foreach \i in {1,2} {
          \foreach \x in {-3, ..., 3} {
            \draw (v\i') -- ($(U\i) + (15*\x:\smlr)$);
            \draw ($(V\i) + (180+15*\x:\bigr)$) -- ($(U\i') + (-15*\x:\bigr)$);
            \foreach \y in {-3, ..., 3} {
              \draw ($(U\i) + (180+15*\x:\smlr)$) -- ($(V\i)+(15*\y:\bigr)$);
            };
          };
        };
        \foreach \x in {-1, ..., 5}
          \draw (v1') -- ($(U2)+(15*\x:\smlr)$);
          
        \node at ($(U1') + (0,\bigr)$) [label=above:$U_1'$] {};
        \node at ($(V1)  + (0,\bigr)$) [label=above:$V_1$]  {};
        \node at ($(U1)  + (0,\smlr)$) [label=above:$U_1$]  {};
        \node at ($(U2') - (0,\bigr)$) [label=below:$U_2'$] {};
        \node at ($(V2)  - (0,\bigr)$) [label=below:$V_2$]  {};
        \node at ($(U2)  - (0,\smlr)$) [label=below:$U_2$]  {};
        
        \draw [->] (3, 1.5) -- (4.5,1.5);
      \end{tikzpicture}
    \caption{An example of the reduction from \apklbis{$(c+1)$} to \apklbis{$c$} used in the proof of Theorem~\ref{thm:k-approx} when $G = P_3$. Each vertex $u_i \in U$ is replaced by three vertex sets $U_i'$, $V_i$ and $U_i$ in the resulting graph $G'$. Note that $G'$ does not depend on the input parameter $\ell$.}\label{fig:k-approx}
  \end{figure}

\begin{thm}\label{thm:k-approx} For all $c \ge 0$, \apkbis{$c$} is \textnormal{NP}-hard.
\end{thm}
\begin{proof}
  
For all $c\geq0$, we give a polynomial-time Turing reduction from the problem
 \apklbis{$(c+1)$} to the problem \apkbis{$c$}. The result then follows from Theorem~\ref{thm:kl-approx}.
  
Fix $c\geq 0$ and let   $(G,\ell)$ be an instance of \apklbis{$(c+1)$}. 
Suppose that $G = (U,V,E)$ where  $U = \{u_1, \dots, u_p\}$. 
Note from the problem definition that $n=|U\cup V|$ and 
suppose without loss of generality that $\ell\in[p]$ and that $n\geq 40$ 
(otherwise, $(G,\ell)$ is an easy instance of \apklbis{$(c+1)$},
so the answer can be computed, even without using the oracle).

Let $s = 2n^6$ and $t = \floor{s\log_23}-s$.
For each $i \in [p]$, let $U_i$, $V_i$ and $U_i'$ be disjoint sets of vertices with $|U_i'| = |V_i| = s$ and $|U_i| = t$. Write $U_i' = \{u_{i,1},\dots,u_{i,s}\}$ and $V_i = \{v_{i,1},\dots,v_{i,s}\}$. Then let  $U' = \bigcup_{i \in [p]} (U_i \cup U_i')$,
$V' = \bigcup_{i \in [p]} V_i \cup V$, and
$$ E' = \bigcup_{i \in [p]} \Big((U_i \times V_i) \cup \{(u_{i,j},v_{i,j}) \mid j \in [s]\}\Big) \cup \bigcup_{(u_j,v) \in E(G)} (U_j \times \{v\}).$$
Let $G' = (U',V',E')$, 
as depicted in Figure~\ref{fig:k-approx}. 

Intuitively, the proof will proceed as follows. We will map independent sets $X'$ of $G'$ to independent sets $X$ of $G$ by taking $X \cap V = X'\cap V$ and adding each $u_i \in U$ to $X$ if and only if $U_i \cap X' \ne \emptyset$. We will show that roughly half the independent sets of each gadget $U_i' \cup V_i \cup U_i$ have this form. We will also show that within each gadget, almost all independent sets with vertices in $U_i$ have size roughly $(s+t)/2$, and almost all others have size roughly $2s/3$. Thus the independent sets in $G$ with $\ell$ vertices in $U$ roughly correspond to the independent sets in $G'$ of size roughly $\ell\cdot(s+t)/2 + (p-\ell)\cdot2s/3$, which we count using a \kbis\ oracle.

We start by defining disjoint sets of independent sets of~$G'$.
For $x \in  \{0,\ldots,p\}$, let
$E(x) = \frac{2s}{3}(p-x) + \frac{s+t}{2}x$
and let 
$$ 
\mathcal{A}_{x} = \Big\{X' \subseteq V(G') \Bigm| \mbox{$X'$ is an independent set of~$G'$ and $\big||X'| - E(x)\big| \le \frac{s}{20}+n$}\Big\}.  
$$
Note that since $n\geq 3$, we have $t> 17s/30$ and $120 n \leq s$.
Thus, if $x'>x$,
$$E(x')-E(x) = \left(\frac{t}{2}-\frac{s}{6}\right)(x'-x) > \left(\frac{17}{60}-\frac{1}{6}\right)s = \frac{s}{10}+\frac{s}{60} \geq\frac{s}{10}+2n.$$
We conclude that the sets $\mathcal{A}_0,\ldots,\mathcal{A}_{p}$ are disjoint. 

Next, we connect the independent sets of~$G'$ with those of~$G$.
Each independent set $X'$ of $G'$ projects onto the independent set $(X' \cap V) \cup \{u_i \mid X' \cap U_i \ne \emptyset\}$ of $G$.
Given an independent set~$X$ of~$G$,
let $\varphi(X)$ be the set of independent sets $X'$ of $G'$ which project onto $X$. If $u_i \in X$, then there are $2^{t}-1$ possibilities for $X' \cap U_i$ and $2^s$ possibilities for $X' \cap U_i'$, but $X' \cap V_i$ is empty. If $u_i \notin X$, then $X' \cap U_i$ is empty and there are $3^s$ possibilities for $X' \cap (U'_i \cup V_i)$. 
For $x\in \{0,\ldots,p\}$, let
$F(x) = (2^{s+t}-2^s)^{x}\cdot3^{(p-x)s}$.
It follows that, 
for any $x$-left independent set~$X$ of~$G$, 
 $|\varphi(X)| = F(x)$, which establishes the first of the following claims.
 
 \begin{description}
 \item Claim~1. For any   $\ell$-left independent set~$X$ of~$G$,
 $|\varphi(X) \cap \mathcal{A}_\ell| \leq F(\ell)$.
 \item Claim~2. For any   $\ell$-left independent set~$X$ of~$G$,
  $|\varphi(X) \cap \mathcal{A}_\ell| \geq F(\ell)/2$.
  \item Claim~3. For any $x\in\{0,\ldots,p\}\setminus \{\ell\}$ 
  and any $x$-left independent set~$X$ of~$G$,
   $|\varphi(X) \cap \mathcal{A}_{\ell} | \leq F(\ell)/2^n$.
 \end{description}
 
The proofs of Claims~2 and~3 are mere calculation, so before proving them we use the claims to complete the
proof of the lemma.
Recall that $(G,\ell)$ is an instance of 
\apklbis{$(c+1)$} with $\ell\in[p]$ and $n\geq 2$.
Together,  the claims imply
\begin{equation}\label{eq:dotwo} \frac{F(\ell)}{2} \LIS{\ell}(G) \leq |\mathcal{A}_\ell| \leq F(\ell) \LIS{\ell}(G) + F(\ell),
\end{equation}
where the final $F(\ell)$ comes from the contribution to $|\mathcal{A}_\ell|$
corresponding to the (at most $2^n$)   independent sets of~$G$  
that are not $\ell$-left independent sets.
Since $\ell\in[p]$, the quantity $\LIS{\ell}(G) $
is at least~$1$, which means that  the right-hand side of~\eqref{eq:dotwo} is at most $2F(\ell)\LIS{\ell}(G)$.
Also,   $F(\ell)>0$. Thus,
  \eqref{eq:dotwo} implies 
$$
\frac  {\LIS{\ell}(G)}{2} \leq \frac{|\mathcal{A}_\ell|}{F(\ell)} \leq  2 \LIS{\ell}(G).$$
The oracle for \apkbis{$c$}
can be
used to 
compute a number $z$ such that
$n^{-c} |\mathcal{A}_\ell| \leq z \leq n^c |\mathcal{A}_\ell|$.
(To do this, just call the oracle repeatedly with input~$G'$
and with every non-negative integer~$k$ such that 
$|k-E(\ell)| \leq \frac{s}{20}+n$, adding the results.)
Thus,
$$
n^{-c} \frac  {\LIS{\ell}(G)}{2} \leq n^{-c} \frac{|\mathcal{A}_\ell|}{F(\ell)} 
\leq \frac{z}{F(\ell)} \leq n^c \frac{|\mathcal{A}_\ell|}{F(\ell)} 
\leq  2 n^c \LIS{\ell}(G),$$
so the desired approximation of  $\LIS{\ell}(G)$ can be achieved by dividing~$z$
by~$F(\ell)$.
We now complete the proof by proving  Claims~2 and~3.
  
\noindent \textbf{Proof of Claim~2:\quad}
Consider any $x\in  \{0,\ldots,p\}$ and let $X$ be an $x$-left independent set of~$G$.
We will show
$|\varphi(X) \cap \mathcal{A}_x| \geq F(x)/2$, which implies the claim by taking $\ell=x$.
In fact, we will establish the much stronger inequality
\begin{equation}
|\varphi(X) \cap \mathcal{A}_x| \geq (1-3ne^{-n^2})F(x),
\label{eq:uselaterJan}
\end{equation} 
which will also be useful in the proof of Claim~3.  
To establish Equation~\eqref{eq:uselaterJan}
we will show that 
the probability that a random element~$Y$ of~$\varphi(X)$
satisfies $\big||Y| - E(x)\big| \le \frac{s}{20}+n$
is at least $1-3ne^{-n^2}$.

So let $Y$ be a uniformly random element of $\varphi(X)$.
We will show that, with probability at least $1-3ne^{-n^2}$, the  following bullet points hold.
\begin{itemize}
\item For all $i\in[p]$ with $u_i\notin X$, we have 
$\left||Y \cap (U_i \cup V_i \cup U'_i) |-\frac{2s}{3}\right| \le \frac{s}{n^2} $, and
\item for all $i\in[p]$ with $u_i \in X$, we have 
$\left||Y \cap (U_i \cup V_i \cup U'_i)| - \frac{s+t}{2} \right| \le \frac{s+t}{n^2}$,
\end{itemize}
Since 
$n\geq 40$, we have
$(p-x) s/n^2 +  x(s+t)/n^2 \leq 2ps/n^2  \leq s/20$ and $|Y \cap V|\leq n$, so the claim follows.
To obtain the desired failure probability, we will show that,  for any $i\in[p]$, the probability that the relevant bullet point fails to hold
is at most $3 e^{-n^2}$ (so the total failure probability is at most $3ne^{-n^2}$, by a union bound).
 
First, consider any $i\in[p]$ with $u_i \notin X$.
In this case, $Y\cap (U_i \cup V_i \cup U'_i)$ is
generated   by including (independently for each $j \in [s]$) 
one of three possibilities:
(i) $u_{i,j}$ but not $v_{i,j}$,
(ii) $v_{i,j}$ but not $u_{i,j}$, or
(iii) neither $u_{i,j}$ nor $v_{i,j}$. Each of the three choices is equally likely.  
Thus $|Y\cap (U_i \cup V_i \cup U'_i)|$ is distributed binomially with mean $2s/3$, so 
by a Chernoff bound (Lemma~\ref{lem:chernoff}),
the probability that the first bullet point fails for~$i$ is at most
$ 2e^{-s/2n^4} < 3 e^{-n^2}$, as desired.

Second, consider any $i\in[p]$ with $u_i \in X$.
In this case, $Y \cap (U_i \cup V_i \cup U'_i)$ is chosen uniformly from all subsets of
$U_i \cup U'_i$ that contain at least one element of~$U_i$.
The total variation distance between the uniform distribution on these subsets and the
uniform distribution on all subsets of $U_i \cup U'_i$ is at most $2^{-t}$.
Also, by a Chernoff bound (Lemma~\ref{lem:chernoff}),
the probability that a uniformly-random subset of $U_i \cup U'_i$ 
has a size that differs from its mean, $(s+t)/2$, by at least $(s+t)/n^2$ is
at most $2e^{-2(s+t)/(3n^4)}$.
Thus, the probability that the second bullet point fails for~$i$ ist
at most $2^{-t} + 2e^{-2(s+t)/(3n^4)}\le 3e^{-n^2}$, as  desired.

\noindent \textbf{Proof of Claim~3:\quad} 
Suppose that $x \in \{0, \dots, p\} \setminus \{\ell\}$ and that $X$ is an $x$-left independent set of~$G$.
We know from Equation~\eqref{eq:uselaterJan}
that 
$|\varphi(X) \cap \mathcal{A}_{\ell}| \leq 3ne^{-n^2}F(x)$.
We wish to show that this is at most $F(\ell)/2^n$.
Note that  $t\geq1$ and $3^{s-1} \le 2^{s+t} \le 3^s$, so
for all 
$y \in  \{0, \dots, p\}$,   
  \begin{align*}
    F(y) &= (2^{s+t}-2^s)^{y}\cdot3^{ps-ys}\le 2^{y(s+t)} \cdot 3^{p s-y s} \le 3^{p s}, \mbox{and}\\
    F(y) &\ge 2^{y(s+t)-y}\cdot 3^{p s-y s} \ge 3^{ps-2y} \geq 3^{ps - 2n}.
  \end{align*}
The claim follows from
$F(x) \leq 3^{ps} \leq   3^{2n} F(\ell)$ and from the fact that $n\ge 40$. \end{proof}
    
\section{Algorithms}

In this final section, we give our algorithmic results: An FPT randomized
approximation scheme (FPTRAS) for $\kbis$, and an exact FPT-algorithm for all
three problems in bounded-degree graphs.
We define an FPTRAS of \kbis\ as in Arvind and Raman~\cite{Arvind2002}.
\begin{defn}\label{def:FPTRAS}
An \emph{FPTRAS} for \kbis\ is a randomised algorithm that takes as input a
bipartite graph $G$, a non-negative integer $k$, and a real number
$\eps\in(0,1)$ and outputs a real number~$z$. With probability at least $2/3$,
the output~$z$ must satisfy $(1-\eps)\IS_k(G) \le z \le (1+\eps)\IS_k(G)$. Furthermore, there is a function $f:\R\rightarrow\R$ and a polynomial $p$ such that the running time of the algorithm is at most $f(k)\,p(|V(G)|,1/\eps)$.
\end{defn}

\begin{thm}\label{thm:kbis-fptras}
  There is an FPTRAS for \kbis\ with time complexity
  $O\left(2^k\cdot k^2/\eps^2\right)$
  for input graphs with~$n$ vertices and~$m$ edges.
\end{thm}
\begin{proof}
Let $(G,k)$ be an instance of \kbis\ with $G = (U,V,E)$ and $n = |V(G)|$.
Let $\varepsilon > 0$ be the other input of the FPTRAS.  
Let $t = 10\ceil{2^k/\eps^2}$. 
The FPTRAS independently samples~$t$ uniformly-random size-$k$ subsets of $U\cup
V$. 
Let $X$ be the number of independent sets among the samples. 
The output $z$ of the FPTRAS is $z=X\cdot\binom{n}{k}/t$.

Note that $\E(X) = t\cdot\IS_k(G)/\binom{n}{k}$.
We now show that with probability at least $2/3$,
$$(1-\eps)\IS_k(G) \le z \le (1+\eps)\IS_k(G).$$
Since each sample lies entirely within~$U$ or entirely within~$V$ with
probability at least $2^{-k}$, we have $\E(X) \ge t2^{-k} \ge 10/\eps^2$.
By Lemma~\ref{lem:chernoff}, we have 
$$
\pr\Big(|X - \E(X)| \ge \eps\E(X)\Big) \le 2e^{-10/3} < 1/3.
$$
Thus, with probability at least $2/3$, we have $|X - \E(X)| \le \eps\E(X)$, and
so $|z - \IS_k(G)| \le \eps\IS_k(G)$ holds as required.

Recall that we use the word-RAM model, in which operations on $O(\log n)$-sized words take $O(1)$ time. Thus for each of the~$t$ samples, the algorithm generates the sample in $O(k)$ time and makes $\binom{k}{2}$ queries to the graph to
check that the selected set is an independent set. 
The running time is therefore as claimed.
\end{proof}

We now turn to our algorithms for bounded-degree graphs. 
We require the following definitions.
For any positive integer $s$, an \emph{$s$-coloured graph} is a tuple $(G,c)$
where $G$ is a graph and $c:V(G) \rightarrow [s]$ is a map. Suppose $\cg{G} =
(G,c)$ and $\cg{G}' = (G',c')$ are coloured graphs with $G = (V,E)$ and $G' = (V',E')$. 

We say a map $\phi:V\rightarrow V'$ is a \emph{homomorphism} from $\cg{G}$ to $\cg{G}'$ if $\phi$ is a homomorphism from $G$ to $G'$ and, for all $v \in V$, $c(v) = c'(\phi(v))$. If $\phi$ is also bijective, we say $\phi$ is an \emph{isomorphism} from $\cg{G}$ to $\cg{G}'$, that $\cg{G}$ and $\cg{G}'$ are \emph{isomorphic}, and write $\cg{G} \simeq \cg{G}'$.
For all $X \subseteq V$, we define $\cg{G}[X] = (G[X], c|_X)$, and say $\cg{G}[X]$ is an \emph{induced subgraph} of $\cg{G}$. Given coloured graphs $\cg{H}$ and~$\cg{G}$, we denote the number of sets $X
\subseteq V(\cg{G})$ with $\cg{G}[X] \simeq \cg{H}$ by $\Ind{\cg{H}}{\cg{G}}$.
Finally, we define $V(\cg{G}) = V$ and $E(\cg{G}) = E$ and we define
$\Delta(\cg{G})$ to be the maximum degree
of~$G$.

For each positive integer~$\Delta$, we consider a counting version of the
induced subgraph isomorphism problem for coloured graphs of degree at
most~$\Delta$. 
\medskip

\myfprob{\sub}
{Two coloured graphs $\cg{H}$ and $\cg{G}$, each with maximum degree bounded by $\Delta$}
{$\Ind{\cg{H}}{\cg{G}}$}
{$|V(\cg{H})|$}

We will later reduce our bipartite independent set counting problems to the
coloured induced subgraph problem.
Note that \sub\ can be expressed as a first-order model-counting problem in
bounded-degree structures.
A well-known result of Frick~\cite[Theorem 6]{Frick2004} would yield an algorithm for
\sub\ with running time $g(k)\cdot n$, where $k=|V(\cg{H})|$ and $n=|V(\cg{G})|$.
(To our knowledge this fact has not appeared in the literature, but the proof is not hard.)
However, the function $g$ of Frick's algorithm may grow faster than any
constant-height tower of exponentials.
In the following, we provide an algorithm for \sub{} that is substantially faster: It runs in time $O(n k^{(2\Delta+3)k})$.

The algorithm follows the strategy of~\cite{hombasis2017} to count small subgraphs: Instead of counting (coloured) induced subgraphs, we can count (coloured) homomorphisms and recover the number of induced subgraphs via a simple basis transformation.
Transforming to homomorphisms is useful because homomorphisms from small patterns to bounded-degree host graphs can be counted by a simple branching procedure--this is however not true for small induced subgraphs. The following lemma encapsulates counting homomorphisms in graphs of bounded degree. Given coloured graphs $\cg{H}$ and $\cg{G}$, we denote the number of homomorphisms from $\cg{H}$ to $\cg{G}$ by $\Hom{\cg{H}}{\cg{G}}$.

\begin{lem}\label{lem:coloured homs in bounded degree}
  There is an algorithm to compute $\Hom{\cg{H}}{\cg{G}}$ in time~$O(n k^k (\Delta+1)^k)$, where~$\cg G$ is a coloured graph with~$n$ vertices, $\cg H$ is a coloured graph with~$k$ vertices, and both graphs have maximum degree at most~$\Delta$.
\end{lem}
\begin{proof}
  The algorithm works as follows:
  If~$\cg H$ is not connected, let~$\cg{H}_1,\dots,\cg{H}_\ell$ be its connected components. Then it is straightforward to verify that
  \begin{equation*}
    \Hom{\cg H}{G}
    =
    \prod_{i=1}^\ell \Hom{\cg{H}_i}{\cg G}\,.
  \end{equation*}
  Thus it remains to describe the algorithm for connected pattern graphs~$\cg H$.

  Let~$\cg H$ be connected.
  A sequence of vertices $v_1, \dots, v_k$ in a graph $F$ is a \emph{traversal} if, for all $i \in \{1,\dots,k-1\}$, the vertex $v_{i+1}$ is contained in $\{v_1,\dots,v_i\} \cup \Gamma(\{v_1, \dots, v_i\})$.
  Let $u_1, \dots, u_k$ be an arbitrary traversal of~$\cg H$ with
  $\{u_1,\dots,u_k\}=V(\cg{H})$; the latter property can be satisfied since~$\cg H$ is a connected graph with~$k$ vertices.
  Note that if $f:V(\cg{H})\rightarrow V(\cg{G})$ is a homomorphism from~$\cg H$ to~$\cg G$, then $f(u_1),\dots,f(u_k)$ is a traversal in $\cg{G}$, and this correspondence is injective.
  Thus the algorithm computes the number of traversals~$v_1,\dots,v_k$ in~$\cg G$ for which the mapping~$f$ with $f(u_i)=v_i$ for all $i$ is a homomorphism from~$\cg H$ to~$\cg G$.
  This number is equal to~$\Hom{\cg{H}}{\cg{G}}$, which the algorithm seeks to compute.

  Since the maximum degree of~$G$ is~$\Delta$, any set~${S \subseteq V(\cg{G})}$ satisfies $|\Gamma(S)| \le \Delta|S|$.
  Thus there are at most $n\cdot(\Delta k + k)^{k-1}$ traversal sequences in $\cg G$, which can be generated in linear time in the number of such sequences.
  For each traversal sequence, verifying whether the sequence corresponds to a homomorphism takes time $O(k\Delta)$ (in the word-RAM model with incidence lists for~$\cg H$ already prepared).
  Overall, we obtain a running time of $O(n\cdot k^k \cdot (\Delta + 1)^k)$.
\end{proof}

Using the above algorithm, we now construct an algorithm that performs a kind of basis transformation to obtain the number of induced coloured subgraphs.
\begin{thm}\label{thm:subbd}
  For all positive integers $\Delta$, there is a fixed-parameter tractable
  algorithm for \sub\ with time complexity
  $O(n \cdot k^{(2\Delta+3)\cdot k})$
  for $n$-vertex coloured graphs $\cg{G}$ and $k$-vertex coloured graphs $\cg{H}$.
\end{thm}
\begin{proof}

  Let $(\cg{H},\cg{G})$ be an instance of \sub, write $\cg{G} = (G,c)$ and $\cg{H} = (H,c')$, and let $k$ be the number of vertices of $\cg{H}$.
  Without loss of generality, suppose that the ranges of $c$ and $c'$ are~$[q]$ for some positive integer $q \le k$.
  Namely, if any vertices of $G$ receive colours not in the range of~$c'$, then our
  algorithm may remove them without affecting $\Ind{\cg{H}}{\cg{G}}$; if any
  vertices of $H$ receive colours not in the range of~$c$, then
  $\Ind{\cg{H}}{\cg{G}} = 0$.

  For coloured graphs $\cg{K}$ and $\cg{B}$, let $\SHom{\cg{K}}{\cg{B}}$ be the number of vertex-surjective homomorphisms from $\cg{K}$ to $\cg{B}$, i.e., the number of those homomorphisms from $\cg{K}$ to $\cg{B}$ that contain all vertices of $\cg{B}$ in their image.
  
  Let $S$ be the set of all $q$-coloured graphs
  $\cg{K}$ such that $\Delta(\cg{K}) \le \Delta$ and, for some $t \in [k]$,
  $V(\cg{K}) = [t]$. Let $S'$ be a set of representatives of (coloured) isomorphism classes of~$S$. 
  
  Let $\boldsymbol{x}$ be the vector indexed by $S'$ such that $\boldsymbol{x}_{\cg{K}} = \Ind{\cg{K}}{\cg{G}}$ for all $\cg{K} \in S'$. This vector contains the number of induced subgraph copies of $\cg{H}$ in $\cg{G}$, but it also contains the number of subgraph copies of all other graphs in $S'$ in $\cg{G}$.
  Let $\boldsymbol{b}$ be the vector indexed by $S'$ such that $\boldsymbol{b}_{\cg{K}} = \Hom{\cg{K}}{\cg{G}}$ for all $\cg{K} \in S'$; each entry of this vector can be computed via the algorithm of Lemma~\ref{lem:coloured homs in bounded degree}. Then we will show that $\boldsymbol{x}$ and $\boldsymbol{b}$ can be related to each other via an invertible matrix $A$ such that $A\boldsymbol{x} = \boldsymbol{b}$. By calculating $A$ and $\boldsymbol{b}$, we can then output $\Ind{\cg{H}}{\cg{G}} = (A^{-1}\boldsymbol{b})_{\cg{H}}$.
    
  To elaborate on this linear relationship between induced subgraph and homomorphism numbers, let us first consider some arbitrary graph $\cg{K} \in S'$. By partitioning the homomorphisms from $\cg{K}$ to $\cg{G}$ according to their image, we have
  \[
  \Hom{\cg{K}}{\cg{G}} = \sum_{\substack{X \subseteq V(\cg{G})\\|X| \le k}}\SHom{\cg{K}}{\cg{G}[X]}.
  \]
  In the right-hand side sum, we can collect terms with isomorphic induced subgraphs $\cg{G}[X]$, since we clearly have $\SHom{\cg{K}}{\cg{B}} = \SHom{\cg{K}}{\cg{B}'}$ if ${\cg{B}} \simeq {\cg{B}'}$. Hence, we obtain
  \begin{equation}\label{eqn:subbd}
  \Hom{\cg{K}}{\cg{G}} = \sum_{\cg{K}'\in S'}\SHom{\cg{K}}{\cg{K}'} \cdot \Ind{\cg{K}'}{\cg{G}}.
  \end{equation}
  Let $A$ be the matrix indexed by $S'$ with $A_{\cg{K},\cg{K}'} = \SHom{\cg{K}}{\cg{K}'}$ for all $\cg{K},\cg{K}' \in S'$. Then \eqref{eqn:subbd} implies that $A\boldsymbol{x} = \boldsymbol{b}$. (An uncoloured version of this linear system is originally due to Lov\'{a}sz~\cite{DBLP:books/daglib/0031021}.)
  
  We next prove that $A$ is invertible. Indeed, given $\cg{K},\cg{K}' \in S'$, write $\cg{K} \lesssim \cg{K}'$ if $\cg{K}$ admits a vertex-surjective homomorphism to $\cg{K}'$. Since $\lesssim$ is a partial order, as is readily verified, it admits a topological ordering $\pi$. Permuting the rows and columns of $A$ to agree with $\pi$ does not affect the rank of $A$, and it yields an upper triangular matrix with non-zero diagonal entries, so it follows that $A$ is invertible.
  
  The algorithm is now immediate. It first determines $S$ by listing all $q$-coloured graphs on at most $k$ vertices with at most $\floor{\Delta k/2}$ edges, then checking each one to see whether it satisfies the degree condition. It then determines $S'$ from $S$ by testing every pair of coloured graphs in~$S$ for isomorphism (by brute force). It then determines each entry $A_{\cg{K},\cg{K}'}$ of $A$ (by brute force) by listing the vertex-surjective maps $\cg{K}\rightarrow\cg{K}'$. It then determines $\boldsymbol{b}$ by invoking Lemma~\ref{lem:coloured homs in bounded degree} to compute each entry $\boldsymbol{b}_{\cg{K}} = \Hom{\cg{K}}{\cg{G}}$ for ${\cg{K}} \in S'$. Finally, it outputs $\Ind{\cg{H}}{\cg{G}} = (A^{-1}\boldsymbol{b})_{\cg{H}}$.
  
  \emph{Running time.} All arithmetic operations are applied to integers bounded by $n^k$, so they each fit into $O(k)$ words, and we bound the complexity of each operation crudely by $O(k^2)$. 
  The number of $q$-coloured graphs on $t$ vertices with at most $\floor{\Delta k/2}$ edges is at most
  \[
  k\cdot q^k \cdot \sum_{m=0}^{\floor{\Delta k/2}} \binom{\binom{k}{2}}{m} \le k \cdot k^k \cdot \frac{\Delta k}{2} \cdot k^{2\floor{\Delta k/2}} 
  = O(k^{2+(\Delta+1)k}) \mbox{ as a function of }k,
  \]
  so our algorithm determines $S$ in time $O(k^{(\Delta+2)k})$ and $|S| = O(k^{2+(\Delta+1)k})$. Moreover, checking whether two graphs in $S$ are isomorphic by brute force requires $O(k^2\cdot k!)$ time, so our algorithm determines $S'$ in time $O(|S|^2k^2\cdot k!) = O(k^{(2\Delta+3)k})$ time. In determining $A$, the algorithm checks at most $k!$ possible homomorphisms for each of $|S'|^2$ pairs of graphs, so it again takes time $O(k^{(2\Delta+3)k})$.
  In determining $\boldsymbol{b}$, the algorithm computes $|S'| = O(k^{2+(\Delta+1)k})$ entries in total, each of which takes time $O(n k^k (\Delta+1)^k)$, so in total it takes time $O(nk^{(\Delta+3)k})$.
  Finally, it takes $O(k^2|S'|^2) = O(k^{(2\Delta+3)k})$ time to invert $A$ and determine $\boldsymbol{x}$ (since $A$ can be put into upper triangular form by permuting rows and columns). Overall, the running time of the algorithm is $O(nk^{(2\Delta+3)k})$, as claimed.
\end{proof}

We note that the above algorithm is not limited to host graphs of bounded degree. 
That is, the same approach can be taken for any host graph class for which counting homomorphisms from (vertex-coloured) patterns with $k$ vertices has an $f(k)\cdot n^{O(1)}$ time algorithm. To this end, simply use this algorithm as a sub-routine instead of Lemma~\ref{lem:coloured homs in bounded degree} in the algorithm constructed in the proof of Theorem~\ref{thm:subbd}.
Examples for such classes of host graphs are planar graphs or, more generally, any graph class of bounded local treewidth \cite{Frick2004}.

\begin{thm}\label{thm:kbisbd}
  For all positive integers $\Delta$:
  \begin{enumerate}[(i)]
    \item $\kbis[\Delta]$ has an algorithm with time complexity $O(|V(G)| \cdot k^{(2\Delta+3)k})$;
    \item $\klbis[\Delta]$ has an algorithm with time complexity $O(|V(G)|\cdot\ell^{\ell(2\Delta^2+8\Delta+4)})$;
    \item $\klmaxbis[\Delta]$ has an algorithm with time complexity $O(|V(G)|\cdot\ell^{\ell(2\Delta^2+8\Delta+4)})$.
  \end{enumerate}
\end{thm}

Recent independent work by Patel and Regts \cite{DBLP:journals/corr/PatelR16} implicitly contains an algorithm for counting independent sets of size $k$ in graphs of maximum degree $\Delta$ in time $O(c^k n)$, where $c$ is a constant depending on $\Delta$. This implies Theorem~\ref{thm:kbisbd}(i). Since our own proof is very short, we provide it for the benefit of the reader. Subsequent work~\cite{PR-new}, published after our original paper~\cite{CDFGL-old}, includes a version of Theorem~\ref{thm:subbd} with running time $\tilde{O}((4\Delta)^{2k}n)$ (which is essentially best-possible under ETH). Note that using this result in the proof of Theorem~\ref{thm:kbisbd}(ii) and (iii) in place of Theorem~\ref{thm:subbd} would not yield algorithms with running times $n\cdot \ell^{o(\ell)}$, as the quantity $|\mathcal{S}_{\ell,r}'|$ defined in the proof is $\ell^{\Omega(\ell)}$ when $\Delta=3$ (for suitable values of $r$).

\begin{proof}
  \textbf{Proof of part (i):} This is immediate from Theorem~\ref{thm:subbd}, since \kbis[$\Delta$] is a special case of \sub\ (taking $\cg{G}$ to be monochromatic and $\cg{H}$ to be a monochromatic independent set of size $k$).
  
  \noindent\medskip \textbf{Proof of part (ii):} For any bipartite graph $G = (U,V,E)$ with degree at most $\Delta$ and any non-negative integers~$\ell$ and~$r$, let $N_{\ell,r}(G)$ be the number of sets $X \subseteq U$ with $|X|=\ell$ and $|\Gamma(X)| = r$. Let $N_{\ell,r}'(G)$ be the number of pairs of sets $X \subseteq U$, $Y \subseteq V$ such that $|X| = \ell$, $|Y| = r$ and $Y \subseteq \Gamma(X)$. Then we have
  \begin{equation}\label{eqn:kbisbd-0}
    N_{\ell,r}(G) = N_{\ell,r}'(G) - \sum_{i=r+1}^{\Delta\ell} \binom{i}{r}N_{\ell,i}(G).
  \end{equation}
  For any bipartite graph $J = (U_J, V_J, E_J)$, we define the corresponding 2-colouring by $c_J(v) = 1$ for all $v \in U_J$ and $c_J(v) = 2$ for all $v \in V_J$. We define the corresponding coloured graph by $\phi(J) = ((U_J \cup V_J, \{\{u,v\} \mid (u,v) \in E_J\}), c_J)$. Let $S_{\ell,r}$ be the set of all bipartite graphs $J = (U_J,V_J,E_J)$ with $U_J = [\ell]$, $V_J = \{\ell+1, \dots, \ell+r\}$, degree at most $\Delta$ and no isolated vertices in $V_J$. Let $\cg{S}_{\ell,r}$ be the corresponding set of coloured graphs, and let $\cg{S}'_{\ell,r}$ be a set of representatives of (coloured) isomorphism classes in $\cg{S}_{\ell,r}$. Then $N_{\ell,r}'(G) = \sum_{\cg{K} \in \cg{S}'_{\ell,r}} \Ind{\cg{K}}{\phi(G)}$, and hence by~\eqref{eqn:kbisbd-0} we have
  \begin{equation}\label{eqn:kbisbd}
  N_{\ell,r}(G) = \sum_{\cg{K} \in \cg{S}'_{\ell,r}} \Ind{\cg{K}}{\phi(G)} - \sum_{i=r+1}^{\Delta\ell} \binom{i}{r}N_{\ell,i}(G).
  \end{equation}
  
  Now suppose that $(G,\ell)$ is an instance of $\klbis[\Delta]$. Then we have
  \begin{equation}\label{eqn:kbisbd-left}
    \LIS{\ell}(G) = \sum_{\substack{X \subseteq U\\|X|=\ell}}2^{|V|-|\Gamma(X)|} = \sum_{r=0}^{\Delta\ell} N_{\ell,r}(G)2^{|V|-r}.
  \end{equation}
  To compute $N_{\ell,\Delta\ell}(G), \dots, N_{\ell,0}(G)$, our algorithm
  applies \eqref{eqn:kbisbd}. For each $r \in \{\Delta\ell, \dots, 0\}$, it
  determines the $\Ind{\cg{K}}{\phi(G)}$ terms of \eqref{eqn:kbisbd} using
  the \sub\ algorithm of Theorem~\ref{thm:subbd}, and the remaining terms of
  \eqref{eqn:kbisbd} recursively with dynamic programming. Finally, it computes
  $\LIS{\ell}(G)$ using~\eqref{eqn:kbisbd-left}. 
  
  To determine the time
  complexity, first note that $|S_{\ell,r}| \le \binom{\Delta\ell^2}{\Delta\ell}
  = O(\ell^{3\Delta\ell})$ holds for all $r \in \{\Delta\ell, \dots, 0\}$. The algorithm therefore determines $\cg{S}_{\ell,r}'$ by brute force in time $O(|S_{\ell,r}|^2(\ell+\Delta\ell)^{\ell+\Delta\ell}) = O(\ell^{\ell(8\Delta+2)})$. The algorithm then calculates each $N_{\ell,r}(G)$ in time 
  \[
  O(|S_{\ell,r}'|\cdot|V(G)| \cdot (\ell+\Delta\ell)^{(2\Delta+3)(\ell+\Delta\ell)}) = O(|V(G)| \cdot \ell^{\ell(2\Delta^2+8\Delta+4)}).
  \]
  The overall running time is therefore $O(|V(G)|\cdot\ell^{\ell(2\Delta^2+8\Delta+4)})$, so part (ii) of the result follows.
  
  \noindent\medskip{Proof of part (iii):} Finally, suppose that $(G,\ell)$ is an instance of $\klmaxbis[\Delta]$. Let $\mu = \min\{r \mid N_{\ell,r}(G) \ne 0\}$, and note that $\MAXLIS{\ell}(G) = N_{\ell,\mu}(G)$. As above, our algorithm determines $N_{\ell,\Delta\ell}(G), \dots, N_{\ell,0}(G)$ using \eqref{eqn:kbisbd}, and thereby determines and outputs $N_{\ell,\mu}(G)$. The overall running time is again $O(|V(G)|\cdot\ell^{\ell(2\Delta^2+8\Delta+4)})$, so part (iii) of the result follows.
\end{proof}

\bibliographystyle{plain}
\bibliography{\jobname}
\end{document}